\let\accentvec\vec

\documentclass[runningheads]{llncs}

\let\vec\accentvec

\usepackage{times}
\usepackage{soul}
\usepackage{url}
\usepackage[utf8]{inputenc}
\usepackage{graphicx}
\usepackage{amsmath}
\usepackage{booktabs}
\usepackage{algorithm}
\usepackage{algorithmic}
\usepackage{mathrsfs}
\usepackage[T1]{fontenc}
\usepackage{verbatim}
\usepackage{color}
\usepackage{stmaryrd}
\usepackage{enumerate}
\usepackage{xspace}
\usepackage{balance}
\usepackage{multirow}
\usepackage{adjustbox}
\usepackage{amssymb}
\usepackage{hyperref}
\usepackage{tikz}
\usetikzlibrary{decorations,decorations.pathmorphing,arrows,shapes,automata,backgrounds,fit,calc,petri,patterns,matrix}
\usepackage{array}
\newcolumntype{H}{>{\setbox0=\hbox\bgroup}c<{\egroup}@{}}

\usepackage{main}

\pgfarrowsdeclare{arrow30}{arrow30}
{
  \arrowsize=0.3pt
  \advance\arrowsize by.275\pgflinewidth
  \pgfarrowsleftextend{+-\arrowsize}
  \advance\arrowsize by.5\pgflinewidth
  \pgfarrowsrightextend{+\arrowsize}
}
{
  \arrowsize=0.3pt
  \advance\arrowsize by.275\pgflinewidth
  \pgfsetdash{}{+0pt}
  \pgfsetroundjoin
  \pgfpathmoveto{\pgfqpoint{1\arrowsize}{0\arrowsize}}
  \pgfpathlineto{\pgfqpoint{-12\arrowsize}{2.5\arrowsize}}
  \pgfpathlineto{\pgfqpoint{-12\arrowsize}{-2.5\arrowsize}}
  \pgfpathclose
  \pgfusepathqfillstroke
}

\newcommand{\compos}     {\mathit{Comp}}

\newcommand{\es}         {\varnothing}
\newcommand{\Nat}        {\mathbb{N}}

\newcommand{\ag}[1]{^{(#1)}}

\renewcommand{\model}{\mathcal{S}}
\newcommand{\assmodel}{\mathcal{A}}

\graphicspath{{./models/}}

\begin{document}

\title{Assume-Guarantee Verification of Strategic Ability} 
\titlerunning{Assume-Guarantee Verification of Strategic Ability}
\author{{\L}ukasz Mikulski\inst{1,2} \and
Wojciech Jamroga\inst{2,3} \and
Damian Kurpiewski\inst{2,1}
}

\institute{
Faculty of Mathematics and Computer Science, Nicolaus Copernicus University, Toru{\'n}, Poland \and
Institute of Computer Science, Polish Academy of Sciences, Warsaw, Poland \and
Interdisciplinary Centre for Security, Reliability and Trust, SnT, University of Luxembourg, Luxembourg}

\maketitle

\begin{abstract}
Model checking of strategic abilities is a notoriously hard problem,
even more so in the realistic case of agents with imperfect information.
Assume-guarantee reasoning can be of great help here, providing a way to
decompose the complex problem into a small set of exponentially easier
subproblems.
In this paper, we propose two schemes for assume-guarantee verification
of alternating-time temporal logic with imperfect information.
We prove the soundness of both schemes, and discuss
their completeness. We illustrate the method by examples based on known
benchmarks, and show experimental results that demonstrate the practical
benefits of the approach.
\end{abstract}


%====================================================================================
\section{Introduction}\label{sec:intro}
%====================================================================================

Multi-agent systems involve a complex network of social and technological components.
Such components often exhibit self-interested, goal-directed behavior, which makes it harder to predict and analyze the dynamics of the system.
In consequence, formal specification and automated verification can be of significant help.

\para{Verification of strategic ability.} Many important properties of multi-agent systems refer to \emph{strategic abilities} of agents and their groups.
\emph{Alter\-nating-time temporal logic} \ATLs~\cite{Alur02ATL,Schobbens04ATL} and \emph{Strategy Logic} \SL~\cite{Mogavero14behavioral} provide powerful tools to reason about such aspects of MAS.
For example, the \ATLs formula $\coop{taxi}\Always\neg\prop{fatality}$ expresses that the autonomous cab can drive in such a way that no one gets ever killed. Similarly, $\coop{taxi,passg}\Sometm\prop{destination}$ says that the cab and the passenger have a joint strategy to arrive at the destination, no matter what the other agents do.
Specifications in agent logics can be used as input to algorithms and tools for \emph{model checking},
that have been in constant development for over 20 years~\cite{Alur98mocha-cav,Busard15reasoning,Cermak15mcmas-sl-one-goal,Huang14symbolic-epist,Kurpiewski21stv-demo,Lomuscio17mcmas}.

Model checking of strategic abilities is hard, both theoretically and in practice. First, it suffers from the well-known state/transition-space explosion.
Moreover, the space of possible strategies is at least exponential \emph{on top of the state-space explosion}, and incremental synthesis of strategies is not possible in general -- especially in the realistic case of agents with partial observability.
Even for the more restricted (and computation-friendly) logic \ATL, model checking of its imperfect information variants is \Deltwo- to \Pspace-complete for agents playing memoryless strategies~\cite{Bulling10verification,Schobbens04ATL} and \EXPTIME-complete to undecidable for agents with perfect recall~\cite{Dima11undecidable,Guelev11atl-distrknowldge}.
The theoretical results concur with outcomes of empirical studies on benchmarks~\cite{Busard15reasoning,Jamroga19fixpApprox-aij,Lomuscio17mcmas}, as well as recent attempts at verification of real-life multi-agent scenarios~\cite{Jamroga20Pret-Uppaal,Kurpiewski19embedded}.

\para{Contribution.}
In this paper, we make the first step towards compositional model checking of strategic properties in asynchronous multi-agent systems with imperfect information.
The idea of \emph{assume-guarantee reasoning}~\cite{Clarke89assGuar,Pnueli84assGuar} is to ``factorize'' the verification task into subtasks where components are verified against a suitable abstraction of the rest of the system.
Thus, instead of searching through the states (and, in our case, strategies) of the huge product of all components, most of the search is performed locally.

To achieve this, we adapt and extend the assume-guarantee framework of~\cite{Lomuscio10assGar,Lomuscio13assGar}.
We redefine the concepts of modules and their composition, follow the idea of expressing assumptions as Büchi automata, and accordingly redefine their interaction with the computations of the coalition.
Then, we propose two alternative assume-guarantee schemes for \ATLs with imperfect information.
The first, simpler one is shown to be sound but incomplete. The more complex one turns out to be both sound and complete.
We illustrate the properties of the schemes on a variant of the Trains, Gate and Controller scenario~\cite{Alur93parametric},
and evaluate the practical gains through verification experiments on models of logistic robots, inspired by~\cite{Kurpiewski19embedded}.

Note that our formal treatment of temporal properties, together with strategic properties of curtailment,\footnote{
  Provided in the supplementary material, available at \href{https://github.com/agrprima22/sup}{https://github.com/agrprima22/sup}. }
substantially extends the applicability of schemes in~\cite{Lomuscio10assGar,Lomuscio13assGar} from temporal liveness properties to strategic properties with arbitrary \LTL objectives.
We also emphasize that our schemes are sound for the model checking of agents with \emph{imperfect} as well as \emph{perfect recall}. In consequence, they can be used to facilitate verification problems with a high degree of hardness, including the undecidable variant for coalitions of agents with memory. In that case, the undecidable problem reduces to multiple instances of the \EXPTIME-complete verification of individual abilities.

\para{Structure of the paper.}
In Section~\ref{sec:models}, we present the model of concurrent MAS that we consider in this paper.
In Section~\ref{sec:logic}, we define the syntax and semantics of the logic used in the formulation of agents' strategic properties.
In Sections~\ref{sec:assumptions} and~\ref{sec:agv-single-agents} we introduce the notions of assumption and guarantee, and utilize them to propose two schemes of assume-guarantee reasoning for strategic abilities.
Finally, we present preliminary results of experimental verification in Section~\ref{sec:experiments} and conclude the paper in Section~\ref{sec:conlusions}.

\para{Related Work.}
Compositional verification (known as \emph{rely-guarantee} in the program verification community) dates back to the early 1970s and the works of Hoare, Owicki, Gries and Jones~\cite{Hoare69axiomatic,Jones83relyGuar,Owicki76relyGuar}.
Assume-guarantee reasoning for temporal specifications was introduced a decade later~\cite{Clarke89assGuar,Pnueli84assGuar}, and has been in development since that time~\cite{Devereux03compositionalRA,Fijalkow20assGuar,Henzinger98assGuar,Kwiatkowska10assGuar,Lomuscio10assGar,Lomuscio13assGar}. Moreover, automated synthesis of assumptions for temporal reasoning has been studied in~\cite{Chen10assGuar-learning,Giannakopoulou05assGuar-learning,He16assGuar-learning,Kong10assGuar-learning}.

The works that come closest to our proposal are~\cite{Devereux03compositionalRA,FinPas,Lomuscio10assGar,Lomuscio13assGar}.
In~\cite{Lomuscio10assGar,Lomuscio13assGar}, models and a reasoning scheme are defined for assume-guarantee verification of liveness properties in distributed systems. We build directly on that approach and extend it to the verification of strategic abilities.
\cite{Devereux03compositionalRA}~studies assume-guarantee reasoning for an early version of \ATL. However, their assume-guarantee rules are designed for perfect information strategies (whereas we tackle the more complex case of imperfect information), and targeted specifically the verification of aspect-oriented programs.
Finally, \cite{FinPas}~investigates the compositional synthesis of strategies for \LTL objectives. The difference to our work is that they focus on finite-memory strategies while we consider the semantics of ability based on memoryless and perfect recall strategies. Another difference lies in our use of repertoire functions that define agents’ choices in a flexible way, and make it closer to real applications.
The advantage of the solution presented in \cite{FinPas}~is the use of contracts, thanks to which it is possible to synthesize individual strategies using the knowledge of the coalition partners' strategies.
We also mention~\cite{ChatHen} that studies the synthesis of Nash equilibrium strategies for 2-player coalitions pursuing $\omega$-regular objectives. The authors call their approach \emph{assume-guarantee strategy synthesis}, but the connection to assume-guarantee verification is rather loose.

A preliminary version of the ideas, presented here, was published in the extended abstract~\cite{MKJ_AG_AAMAS22}.
Our extension of the STV tool~\cite{Kurpiewski21stv-demo}, used in the experiments, is described in the companion paper~\cite{Kurpiewski22STV+AGV}.

%====================================================================================
\section{Models of Concurrent MAS}\label{sec:models}
%====================================================================================

Asynchronous MAS have been modeled by variants of reactive modules~\cite{Alur99reactive,Lomuscio13assGar} and automata networks~\cite{Jamroga21paradoxes-kr}.
Here, we adapt the variant of reactive modules that was used to define assume-guarantee verification for temporal properties in~\cite{Lomuscio13assGar}.

\subsection{Modules}\label{sec:modules}

Let $D$ be the shared domain of values for all the variables in the system.
$D^X$ is the set of all valuations for a set of variables $X$.
The \emph{system} consists of a number of \emph{agents}, each represented by its \emph{module} and a \emph{repertoire} of available choices.
Every agent uses \emph{state variables} and \emph{input variables}.
It can read and modify its state variables at any moment, and their valuation is determined by the current state of the agent.
The input variables are not a part of the state, but their values influence
transitions that can be executed.

\begin{definition}[Module~\cite{Lomuscio13assGar}]\label{d:module}
A \emph{module} is a tuple $M=(X,I,Q,T,\lambda,q_0)$, where:\
 $X$ is a finite set of state variables;
 $I$ is a finite set of input variables with $X\cap I=\es$;
 $Q=\{q_0,q_1,\ldots,q_n\}$ is a finite set of states;
 $q_0\in Q$ is an initial state;
 $\lambda:Q\to D^X$ labels each state with a valuation of the state variables;
 finally, $T\subseteq Q\times D^I\times Q$ is a transition relation such that
  (a) for each pair $(q,\alpha)\in Q\times D^I$ there exists $q'\in Q$ with $(q,\alpha,q')\in T$, and
  (b) $(q,\alpha,q')\in T, q\neq q'$ implies $(q,\alpha,q)\notin T$.
In what follows, we omit the self-loops from the presentation.
\end{definition}

Modules $M,M'$ are \emph{asynchronous} if $X\cap X'=\es$.
We extend modules by adding \emph{repertoire functions} that define the agents' available choices in a way similar to~\cite{Jamroga21paradoxes-kr}. 

\begin{definition}[Repertoire]
Let $M=(X,I,Q,T,\lambda,q_0)$ be a module of agent $i$.
The \emph{repertoire} of $i$ is defined as $R:Q\to \powerset{\powerset{T}}$, i.e., a mapping from local states to \emph{sets of sets of transitions}.
Each $R(q) = \{T_1,\dots,T_m\}$ must be nonempty and consist of nonempty sets $T_i$ of transitions starting in $q$.
If the agent chooses $T_i \in R(q)$, then only a transition in $T_i$ can be occur at $q$ within the module.
\end{definition}

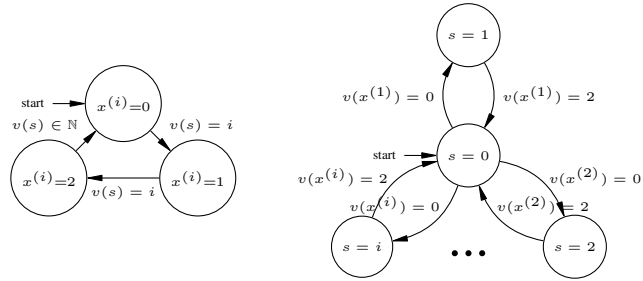
\begin{figure}[t]\centering
\begin{tabular}{@{}c@{\qquad}c@{}}
\begin{tabular}{@{}c@{}}
\begin{tikzpicture}[->,auto,>=arrow30,node distance=1.4cm,font=\tiny]
  \node[initial,state] (A)                    {$x\ag{i}\!\!=\!\!0$};
  \node[state]         (B) [below right of=A] {$x\ag{i}\!\!=\!\!1$};
  \node[state]         (C) [below left of=A]  {$x\ag{i}\!\!=\!\!2$};
  
  \path (A) edge node {$v(s)=i$} (B)
        (B) edge node {$v(s)=i$} (C)
        (C) edge node {$v(s)\in\Nat$} (A);  
\end{tikzpicture}
\end{tabular}
  &
\begin{tabular}{@{}c@{}}
\begin{tikzpicture}[->,auto,>=arrow30,node distance=2cm,font=\tiny]
  \node[initial,state] (A)                    {$s=0$};
  \node[state]         (B) [above of=A, yshift = -4mm]       {$s=1$};
  \node[state]         (C) [below left of=A, yshift = 2mm]  {$s=i$};
  \node[state]         (D) [below right of=A, yshift = 2mm] {$s=2$};
  \node                (E) [below of=A, yshift = 7mm] {{\huge ...}};
  
  \path (A) edge [bend left] node {$v(x\ag{1})=0$} (B)
            edge [bend left] node [swap, xshift=2mm] {$v(x\ag{i})=0$} (C)
            edge [bend left] node [yshift=-2mm] {$v(x\ag{2})=0$} (D)
        (B) edge [bend left] node {$v(x\ag{1})=2$} (A)
        (C) edge [bend left] node [yshift=-2mm] {$v(x\ag{i})=2$} (A)
        (D) edge [bend left] node [xshift=-2mm] [swap] {$v(x\ag{2})=2$} (A);  
\end{tikzpicture}
\end{tabular}
\end{tabular}
\caption{A variant of TCG: Train synchronizing with a semaphore (left) and the controller (right).}
\label{fig:train2}
\label{fig:controller}
\vspace*{-2mm}
\end{figure}

We
adapt
the Train-Gate-Controller (TGC) benchmark~\cite{Alur98mocha-cav} as our running example.

\begin{example}\label{ex1:train}
The module $M\ag{i}$ of a train is presented in Figure~\ref{fig:train2}~(left).
Its local states $Q\ag{i}=\{w\ag{i},t\ag{i},a\ag{i}\}$ refer, respectively, to the train \textbf{w}aiting at the entrance, riding in the \textbf{t}unnel, and cruising \textbf{a}way from the tunnel.
The sole state variable $x\ag{i}$ labels the state with values $0$, $1$, and $2$, respectively.
$I\ag{i}=\{s\}$ consists of a single input variable that takes values from an external multi-valued semaphore.
The train can enter and exit the tunnel only if the semaphore allows for that, i.e., if $v(s)=i$.
To this end, we define
$T\ag{i}=\{(w\ag{i},i,t\ag{i}),(t\ag{i},i,a\ag{i}),(a\ag{i},0,w\ag{i}),(a\ag{i},1,w\ag{i})\ldots,$
$(a\ag{i},n,w\ag{i})\} \cup \{(w\ag{i},j,w\ag{i}),(t\ag{i},j,t\ag{i}) \mid j\neq i\}$
.\footnote{
  By a slight abuse of notation, the valuation of a single variable is identified with its value. }

The module $M\ag{C(n)}$ of a controller that coordinates up to $n$ trains is depicted in Figure~\ref{fig:controller}~(right).
  Formally, it is defined by: 
\begin{itemize}
\item $X = \{s\}$ (the semaphore),
\item $I=\{x_1,\dots,x_n\}$ (the positions of trains),
\item $Q=\{r,g_1,\ldots,g_n\}$ (red or directed green light),
\end{itemize}  
  where a state with subscript $1$ represents a tunnel shared with the other trains,
  $\lambda(g_i)(s)=i$, $\lambda(r)(s)=0$, and $r$ is the initial state.
  
  The controller can change the light to green when a train is waiting for the permission to enter the tunnel, and back to red after it passed through the tunnel: $T=\{(r,v,g_i) \mid v(x_i)=0\}\cup\{(g_i,v,r) \mid v(x_i)=2\}$.

Each agent can freely choose the local transition intended to execute next. Thus,
$R\ag{i}(q) = \set{\set{(q,\alpha,q')} \mid (q,\alpha,q')\in T\ag{i}}$, and similarly for $R\ag{C(n)}$.

Note that all the modules in TCG are asynchronous.
\end{example}

\subsection{Composition of Agents}\label{sec:composition}\label{sec:models-properties}

On the level of the temporal structure, the model of a multi-agent system is given by the asynchronous composition $M = M\ag{1} | \dots | M\ag{n}$ that combines modules
$M\ag{i}$
into a single module.
The definition is almost the same as in~\cite{Lomuscio13assGar}; we only extend it to handle the repertoire functions that are needed to characterize strategies and strategic abilities.

We begin with the notion of compatible valuations to adjust local states of one agent
with the labels of the actions performed by the other agent.
Note that the local states of different asynchronous agents rely on disjoint sets of variables.

Let $Y,Z\subseteq X$ and $\rho_1\in D^Y$ while $\rho_2\in D^Z$.
We say that $\rho_1$ is compatible with $\rho_2$ (denoted by $\rho_1 \sim \rho_2$)
if for any $x\in Y\cap Z$ we have $\rho_1(x)=\rho_2(x)$.
We can compute the union of $\rho_1$ with $\rho_2$ which is compatible with $\rho_1$ by setting
$(\rho_1 \cup \rho_2)(x) = \rho_1(x)$ for $x\in Y$ and
$(\rho_1 \cup \rho_2)(x) = \rho_2(x)$ for $x\in Z$.

\begin{definition}[Composition of modules~\cite{Lomuscio13assGar}]\label{d:comp}
The composition of asynchronous modules
$M\ag{1}=(X\ag{1},I\ag{1},Q\ag{1},T\ag{1},\lambda\ag{1},q\ag{1}_{0})$
and $M\ag{2}=(X\ag{2},I\ag{2},Q\ag{2},T\ag{2},\lambda\ag{2},q\ag{2}_{0})$
(with $X\ag{1}\cap X\ag{2}=\es$) is
a composite module
$M=(X=X\ag{1}\uplus X\ag{2},I=(I\ag{1}\cup I\ag{2})\setminus X,Q\ag{1}\times Q\ag{2},T,\lambda,q_0=(q\ag{1}_{0},q\ag{2}_{0}))$,
where
\begin{itemize}
\item $\lambda:Q\ag{1}\times Q\ag{2}\to D^X$, $\lambda(q\ag{1},q\ag{2})=\lambda\ag{1}(q\ag{1})\cup\lambda\ag{2}(q\ag{2})$,
\item $T$ is the minimal transition relation derived by the set of rules presented below:
\[
\bf{ASYN_L}\;\;\;
\begin{array}{c}
q\ag{1}\xrightarrow[]{\alpha\ag{1}}_{T\ag{1}}{q'}\ag{1}\;\;\;\;q\ag{2}\xrightarrow[]{\alpha\ag{2}}_{T\ag{2}}{q'}\ag{2}\\
\alpha\ag{1}\sim\alpha\ag{2} \;\;\;\; \lambda\ag{1}(q\ag{1})\sim\alpha\ag{2} \;\;\;\; \lambda\ag{2}(q\ag{2})\sim\alpha\ag{1}\\
\hline
(q\ag{1},q\ag{2})\xrightarrow[]{(\alpha\ag{1}\cup\alpha\ag{2})\setminus X}_T({q'}\ag{1},q\ag{2})
\end{array}
\]
\[
\bf{ASYN_R}\;\;\;
\begin{array}{c}
q\ag{1}\xrightarrow[]{\alpha\ag{1}}_{T\ag{1}}{q'}\ag{1}\;\;\;\;q\ag{2}\xrightarrow[]{\alpha\ag{2}}_{T\ag{2}}{q'}\ag{2}\\
\alpha\ag{1}\sim\alpha\ag{2} \;\;\;\; \lambda\ag{1}(q\ag{1})\sim\alpha\ag{2} \;\;\;\; \lambda\ag{2}(q\ag{2})\sim\alpha\ag{1}\\
\hline
(q\ag{1},q\ag{2})\xrightarrow[]{(\alpha\ag{1}\cup\alpha\ag{2})\setminus X}_T(q\ag{1},{q'}\ag{2})
\end{array}
\]
\[
\;\;\;\;\;\bf{SYN}\;\;\;
\begin{array}{c}
q\ag{1}\xrightarrow[]{\alpha\ag{1}}_{T\ag{1}}{q'}\ag{1}\;\;\;\;q\ag{2}\xrightarrow[]{\alpha\ag{2}}_{T\ag{2}}{q'}\ag{2}\\
\alpha\ag{1}\sim\alpha\ag{2} \;\;\;\; \lambda\ag{1}(q\ag{1})\sim\alpha\ag{2} \;\;\;\; \lambda\ag{2}(q\ag{2})\sim\alpha\ag{1}\\
\hline
(q\ag{1},q\ag{2})\xrightarrow[]{(\alpha\ag{1}\cup\alpha\ag{2})\setminus X}_T({q'}\ag{1},{q'}\ag{2})
\end{array}
\]
\end{itemize}
pruned in order to avoid disallowed self-loops.
We use the notation $M=M\ag{1}|M\ag{2}$.
\end{definition}

Note that the operation is defined in~\cite{Lomuscio13assGar} for a pair of modules only.
It can be easily extended to a larger number of pairwise asynchronous modules. Moreover, the order of the composition does not matter.

Consider agents $(M\ag{1},R\ag{1}),\ \dots,\ (M\ag{n},R\ag{n})$.
The \emph{multi-agent system} is defined by $\model = (M\ag{1}|M\ag{2}|\ldots|M\ag{n},\ R\ag{1},\ldots, R\ag{n})$, i.e., the composition of the underlying modules, together with the agents' repertoires of choices.

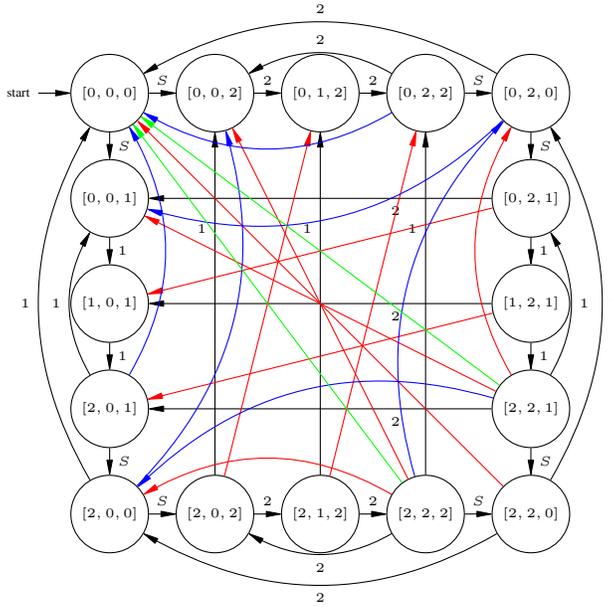
\begin{figure}[t]\centering
\begin{tikzpicture}[->,auto,>=arrow30,node distance=1.4cm,font=\tiny]
  \node[initial,state] (A1)                   {\tiny $[0,0,0]$};
  \node[state]         (B1) [below of=A1]     {$[0,0,1]$};
  \node[state]         (B2) [below of=B1]     {$[1,0,1]$};
  \node[state]         (B3) [below of=B2]     {$[2,0,1]$};
  \node[state]         (A2) [below of=B3]     {$[2,0,0]$};
  \node[state]         (C1) [right of=A1]     {$[0,0,2]$};
  \node[state]         (C2) [right of=C1]     {$[0,1,2]$};
  \node[state]         (C3) [right of=C2]     {$[0,2,2]$};
  \node[state]         (A3) [right of=C3]     {$[0,2,0]$};
  \node[state]         (D1) [below of=A3]     {$[0,2,1]$};
  \node[state]         (D2) [below of=D1]     {$[1,2,1]$};
  \node[state]         (D3) [below of=D2]     {$[2,2,1]$};
  \node[state]         (A4) [below of=D3]     {$[2,2,0]$};
  \node[state]         (E1) [right of=A2]     {$[2,0,2]$};
  \node[state]         (E2) [right of=E1]     {$[2,1,2]$};
  \node[state]         (E3) [right of=E2]     {$[2,2,2]$};
  
  \path (A1) edge node {$S$} (B1)
        (B1) edge node {$1$} (B2)
        (B2) edge node {$1$} (B3)
        (B3) edge node {$S$} (A2)
        (A2) edge [bend left] node {$1$} (A1)
        (B3) edge [bend left] node {$1$} (B1)
        (A1) edge node {$S$} (C1)
        (C1) edge node {$2$} (C2)
        (C2) edge node {$2$} (C3)
        (C3) edge node {$S$} (A3)
        (A3) edge [bend right] node [swap] {$2$} (A1)
        (C3) edge [bend right] node [swap] {$2$} (C1)
        (A3) edge node {$S$} (D1)
        (D1) edge node {$1$} (D2)
        (D2) edge node {$1$} (D3)
        (D3) edge node {$S$} (A4)
        (A4) edge [bend right] node [swap] {$1$} (A3)
        (D3) edge [bend right] node [swap] {$1$} (D1)
        (A2) edge node {$S$} (E1)
        (E1) edge node {$2$} (E2)
        (E2) edge node {$2$} (E3)
        (E3) edge node {$S$} (A4)
        (A4) edge [bend left] node {$2$} (A2)
        (E3) edge [bend left] node {$2$} (E1)
        (E1) edge node [yshift=1cm] {$1$} (C1)
        (E2) edge node [yshift=1cm] {$1$} (C2)
        (E3) edge node [yshift=1cm] {$1$} (C3)
        (D1) edge node [xshift=1cm] {$2$} (B1)
        (D2) edge node [xshift=1cm] {$2$} (B2)
        (D3) edge node [xshift=1cm] {$2$} (B3)
        (A4) edge [color=red] (A1)
        (E1) edge [color=red] (C2)
        (E2) edge [color=red] (C3)
        (E3) edge [color=red] (C1)
        (D1) edge [color=red] (B2)
        (D2) edge [color=red] (B3)
        (D3) edge [color=red] (B1)
        (E3) edge [color=red, bend right] (A2)
        (D3) edge [color=red, bend left] (A3)
        (A2) edge [color=blue, bend right] (C1)
        (E3) edge [color=blue, bend left] (A3)
        (A3) edge [color=blue, bend left] (B1)
        (D3) edge [color=blue, bend right] (A2)
        (E3) edge [color=green] (A1)
        (D3) edge [color=green] (A1)
        (B3) edge [color=blue, bend right] (A1)
        (C3) edge [color=blue, bend left] (A1)
        ;  
\end{tikzpicture}
\caption{Composition of modules: two trains $M\ag{1}, M\ag{2}$ and controller $M\ag{C(2)}$}
\label{fig:trains-composition}
\end{figure}

\begin{example}\label{ex:composition}
The composition $M\ag{1} | M\ag{2} | M\ag{C(2)}$ of two train modules $M\ag{1}, M\ag{2}$ and controller $M\ag{C(2)}$ is presented in Figure~\ref{fig:trains-composition}.
The asynchronous transitions are labelled by the agent performing the transitions.
All the synchronous transitions performed by both trains are in red, while the synchronous transitions performed by a controller with one of the trains are in blue.
There are two synchronous transition performed by all the agents, both in green.
\end{example}

\paragraph{Traces and Words. }\label{sec:traces}
A trace of a module $M$ is an infinite sequence of alternating states and transitions $\sigma=q_0\alpha_0 q_1\alpha_1\ldots$, where
$(q_i,\alpha_i,q_{i+1})\in T$ for every $i\in\Nat$ (note that $q_0$ is the initial state).
An infinite word $w=v_0 v_1,\ldots \in(D^X)^\omega$ is \emph{derived} by $M$ with trace $\sigma=q_0\alpha_0 q_1\alpha_1\ldots$ if $v_i = \lambda(q_i)$ for all $i\in\Nat$.
An infinite word $u=\alpha_0 \alpha_1,\ldots \in(D^I)^\omega$ is \emph{admitted} by $M$ with $\sigma$ if $\sigma=q_0\alpha_0 q_1\alpha_1\ldots$.
Finally, $w$ (resp.~$u$) is derived (resp.~admitted) by $M$ if there exists a trace of $M$ that derives (resp.~admits) it.

%====================================================================================
\section{What Agents Can Achieve}\label{sec:logic}
%====================================================================================

\emph{Alternating-time temporal logic} \ATLs~\cite{Alur02ATL,Schobbens04ATL}
introduces \emph{strategic modalities} $\coop{C}\gamma$, expressing that coalition $C$ can enforce the temporal property $\gamma$.
We use the semantics based on \emph{imperfect information strategies} with \emph{imperfect recall} ($\ir$) or \emph{perfect recall} ($\iR$)~\cite{Schobbens04ATL}.
Moreover, we only consider formulas without the next step operator $\Next$ due to its questionable interpretation for asynchronous systems, which are based on the notion of local clocks.

\para{Syntax.}
Formally, the syntax of $\ATLsX$ is as follows:
\begin{displaymath}
\phi ::= p(Y) \mid \neg\phi \mid \phi\land\phi \mid \coop{C} \gamma\, , \qquad\qquad
\gamma ::= \phi \mid \neg\gamma \mid \gamma\land\gamma \mid \gamma \Until \gamma
\end{displaymath}
where $p:Y\to D$ for some subset of domain variables $Y\subseteq X$. That is, each atomic statement refers to the valuation of variables used in the system.
$\Until$ is the ``strong until'' operator of \LTLX. The ``sometime'' and ``always'' operators can be defined as usual by $\Sometm\gamma \equiv \top \Until \gamma$ and $\Always\gamma \equiv \neg\Sometm\neg\gamma$.
The set of variables used by the formula $\gamma$ is denoted by $var(\gamma)$.

In most of the paper, we focus on formulas that consist of a single strategic modality followed by an \LTLX formula (i.e., $\coop{C}\gamma$, where $\gamma\in\LTLX$).
The corresponding fragment of \ATLsX, called \oneATLsX, suffices to express many interesting specifications, namely the ones that refer to agents' ability of enforcing trace properties (such as safety or reachability of a winning state).
Note that \oneATLsX has strictly higher expressive and distinguishing power than \LTLX. In fact, model checking \oneATLsX is equivalent to \LTLX controller synthesis, i.e., a variant of \LTL realizability.

Nested strategic modalities might be sometimes needed to refer to an agent’s ability to endow or deprive another agent with/of ability. 
We discuss assume-guarantee verification for such specifications in Section~\ref{sec:nested}.

\para{Strategies and Their Outcomes.}
Let $\model$ be a system composed of $n$ agents with asynchronous modules $M\ag{i} = (X\ag{i},I\ag{i},Q\ag{i},T\ag{i},\lambda\ag{i},q_{0}\ag{i})$ and repertoires $R\ag{i}$.

\begin{definition}[Strategies]
A \emph{memoryless strategy} for agent $i$ (\ir-strategy in short) is a function
$s_i^{\ir}:Q\ag{i}\to \powerset{\powerset{T\ag{i}}}$ such that $s_i^{\ir}(q\ag{i})\in R\ag{i}(q\ag{i})$ for every $q\ag{i}\in Q\ag{i}$.
That is, a memoryless strategy assigns a legitimate choice to each local state of $i$.

A \emph{perfect recall strategy} for $i$ (\iR-strategy in short) is a function $s_i^{\iR}:(Q\ag{i})^+\to T\ag{i}$ such that $s_i^{\iR}(q\ag{i}_1,\dots,q\ag{i}_k)\in R\ag{i}(q\ag{i}_k)$, i.e., it assigns choices to finite sequences of local states.
We assume that $s_i^{\iR}$ is stuttering-invariant, i.e., 
$$s_i^{\iR}(q\ag{i}_1,\dots,q\ag{i}_j,q\ag{i}_j,\dots,q\ag{i}_k) = s_i^{\iR}(q\ag{i}_1,\dots,q\ag{i}_j,\dots,q\ag{i}_k).$$
Note that the agent's choices in a strategy depend only on its \emph{local} states, thus being uniform by construction.

\end{definition}

Let $\sigma=q_0\alpha_0 q_1\alpha_1\ldots$ be a trace, where $q_j=(q_j\ag{1},q_j\ag{2},\ldots,q_j\ag{n})$ are global states in $Q\ag{1}\times\ldots\times Q\ag{n}$.
We say that $\sigma$ \emph{implements} strategy $s_i^{\ir}$
if, for any $j$ where $q_j\ag{i}\neq q_{j+1}\ag{i}$, we have $(q_j\ag{i},\alpha_j,q_{j+1}\ag{i})\in s_i^{\ir}(q_j\ag{i})$
where $\alpha_j:I\ag{i}\to D$ and $\alpha_j(x)=\lambda(q_j)(x)$.
A word $w=v_0v_1\dots$ \emph{implements} $s_i^{\ir}$
if it is derived by $\model$ with some trace $\sigma$ implementing $s_i^{\ir}$.
The definitions for $s_i^{\iR}$ are analogous.

\begin{definition}[Coalitional strategies]
Let $C\subseteq \{1,\ldots,n\}$ be a coalition of agents.
A \emph{joint memoryless strategy $s_C^{\ir}$} for $C$ is a collection of memoryless strategies $s_i^\ir$, one per $i\in C$.
We say that a trace $\sigma$ (respectively a word $w_\sigma$) \emph{implements} $s_C^{\ir}$ if it implements every strategy $s_i^{\ir}, i\in C$.
The definitions for joint perfect recall strategies are analogous.
Whenever a claim holds for both types of strategies, we will refer to them simply as ``strategies.''
\end{definition}

\para{Semantics.}
Let $x\in\set{\ir,\iR}$ be a strategy type. The semantics of \ATLsX is given below (we omit the standard clauses for Boolean operators etc.). By $w[i]$, we denote the $i$th item of sequence $w$, starting from $0$.

\begin{description2}
\item[{$\model,q \satisf[x] p(Y)$}] if $\lambda(q)|_Y = p(Y)$;

\item[{$\model,q \satisf[x] \coop{C} \gamma$}] if there exists an $x$-strategy $s_C$ for $C$ such that, for any word $w$ starting in $q$ that implements $s_C$, we have $\model,w \satisf \gamma$;

\item[{$\model,w \satisf \phi$}] if $\model,w[0] \models \phi$;

\item[{$\model,w \satisf \gamma_1 \Until \gamma_2$}] if there exists $j$ such that $\model,w[j,\infty] \satisf \gamma_2$, and $\model,w[i,\infty] \satisf \gamma_1$ for each $0\le i<j$.
\end{description2}

Finally, we say that $\model \satisf[x] \phi$ if $\model,q_0 \satisf[x] \phi$, where $q_0$ is the initial state of $\model$.

\begin{example}\label{ex:formulas}\label{ex:logic}
Let us consider the system $\model$ of Example~\ref{ex:composition}
and the \oneATLs formula
$\phi \equiv \coop{1,2} (GF p\ag{1}\land GF p\ag{2})$,
where $p\ag{i}(x\ag{i})=1$.
That is, $\phi$ says that trains $1,2$ have a strategy so that each visits the tunnel infinitely many times.
Consider the joint strategy $(\sigma_1,\sigma_2)$ with
$\sigma_i(w\ag{i})=\{(w\ag{i},i,T\ag{i})\}$,
$\sigma_i(T\ag{i})=\{(T\ag{i},i,A\ag{i})\}$, and
$\sigma_i(A\ag{i})=\{(A\ag{i},3-i,w\ag{i})\}$.
All the traces implementing $(\sigma_1,\sigma_2)$ alternate the visits of the trains in the tunnel,
making the \LTL formula $GF p\ag{1}\land GF p\ag{2}$ satisfied.
Thus, $\model \satisf[x] \phi$ for $x\in\set{\ir,\iR}$.

By the same strategy, we get 
$\model \satisf[x] \coop{1,2} (GF q\ag{1}\land GF q\ag{2})$, where $q\ag{i}(s)=i$.
\end{example}

%====================================================================================
\section{Assumptions and Guarantees}\label{sec:assumptions}
%====================================================================================

Our assume-guarantee scheme reduces the complexity of model checking by ``factorizing'' the task into verification of strategies of single agents with respect to abstractions of the rest of the system.
In this section, we formalize the notions of \emph{assumption} and \emph{guarantee}, which provide the abstractions in a way that allows for simulating the global behavior of the system.

\subsection{Assumptions}\label{sec:ass}

\begin{definition}[Assumption~\cite{Lomuscio13assGar}]
An \emph{assumption} or an \emph{extended module} $(M,F)=(X,I,Q,T,\lambda,q_0,F)$ is a module augmented with a set of accepting states $F\subseteq Q$.
\end{definition}
For assumptions, we use B\"{u}chi accepting conditions.
More precisely, the infinite word $w=q_0 q_1,\ldots$ is \emph{accepted}
by extended module $(M,F)$ with computation $u=\alpha_0\alpha_1\ldots$
if it is derived by $M$ with a trace $\sigma=q_0\alpha_0 q_1\alpha_1\ldots$ and
$\mathit{inf}(\sigma)\cap F\neq\es$.
Thus, the assumptions have the expressive power of $\omega$-regular languages.
In practical applications, it might be convenient to formulate actual assumptions in \LTL (which covers a proper subclass of $\omega$-regular properties).

The definitions of Sections~\ref{sec:models} and~\ref{sec:logic} generalize to assumptions in a straightforward way.
In particular, we can compose a module $M$ with an assumption $A'=(M',F')$, and obtain an extended composite module $A=(M|M',F)$,
where $F = \{(q,q')\in Q\times Q' \mid q'\in F'\}$.
We use the notation $A=M|A'$.
Moreover, let $\assmodel = (A,R\ag{1},\dots,R\ag{m})$ be a MAS based on the extended module $A$
with repertoires related to all components of $M$.
The semantics of \oneATLsX extends naturally:

\begin{description2}
\item[{$\assmodel,q \satisf[x] \coop{C} \phi$}] iff there exists an $x$-strategy $s_C$ for $C$ such that,
for any word $w=w[1]w[2]\ldots$ that implements $s_C$ and is accepted by $A$, we have $\assmodel,w \satisf[x] \phi$.
\end{description2}

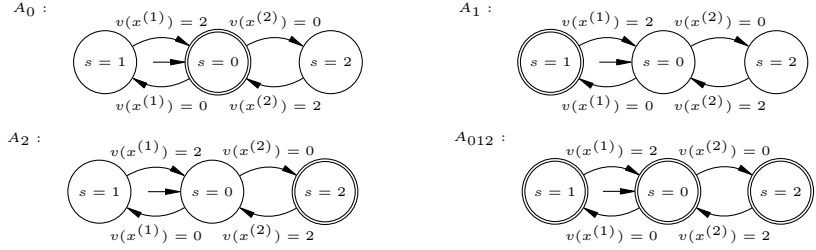
\begin{figure}[t]\centering
  \begin{center}
    \begin{tikzpicture}[->,auto,>=arrow30,node distance=1cm,font=\tiny]
    
    \node[initial,state,accepting,initial text=] (A)                           {$s=0$};
    \node[state]                                 (B) [left of=A, xshift=-5mm]  {$s=1$};
    \node[state]                                 (C) [right of=A, xshift=5mm]  {$s=2$};
    \node (X) [left of=B, yshift=7mm] {$A_0:$};
    
    \path (A) edge [bend left] node {$v(x\ag{1})=0$} (B)
              edge [bend left] node {$v(x\ag{2})=0$} (C)
          (B) edge [bend left] node {$v(x\ag{1})=2$} (A)
          (C) edge [bend left] node {$v(x\ag{2})=2$} (A);
    
    \end{tikzpicture}
    \hspace{1cm}
    \begin{tikzpicture}[->,auto,>=arrow30,node distance=1cm,font=\tiny]
    
    \node[initial,state,initial text=] (A)                           {$s=0$};
    \node[state,accepting]             (B) [left of=A, xshift=-5mm]  {$s=1$};
    \node[state]                       (C) [right of=A, xshift=5mm]  {$s=2$};
    \node (X) [left of=B, yshift=7mm] {$A_1:$};
    
    \path (A) edge [bend left] node {$v(x\ag{1})=0$} (B)
              edge [bend left] node {$v(x\ag{2})=0$} (C)
          (B) edge [bend left] node {$v(x\ag{1})=2$} (A)
          (C) edge [bend left] node {$v(x\ag{2})=2$} (A);
    
    \end{tikzpicture}
    \end{center}
    \vspace*{-7mm}
    \begin{center}
    \begin{tikzpicture}[->,auto,>=arrow30,node distance=1cm,font=\tiny]
    
    \node[initial,state,initial text=] (A)                           {$s=0$};
    \node[state]                       (B) [left of=A, xshift=-5mm]  {$s=1$};
    \node[state,accepting]             (C) [right of=A, xshift=5mm]  {$s=2$};
    \node (X) [left of=B, yshift=7mm] {$A_2:$};
    
    \path (A) edge [bend left] node {$v(x\ag{1})=0$} (B)
              edge [bend left] node {$v(x\ag{2})=0$} (C)
          (B) edge [bend left] node {$v(x\ag{1})=2$} (A)
          (C) edge [bend left] node {$v(x\ag{2})=2$} (A);
    
    \end{tikzpicture}
    \hspace{1cm}
    \begin{tikzpicture}[->,auto,>=arrow30,node distance=1cm,font=\tiny]
    
    \node[initial,state,accepting,initial text=] (A)                           {$s=0$};
    \node[state,accepting]                       (B) [left of=A, xshift=-5mm]  {$s=1$};
    \node[state,accepting]                       (C) [right of=A, xshift=5mm]  {$s=2$};
    \node (X) [left of=B, yshift=7mm] {$A_{012}:$};
    
    \path (A) edge [bend left] node {$v(x\ag{1})=0$} (B)
              edge [bend left] node {$v(x\ag{2})=0$} (C)
          (B) edge [bend left] node {$v(x\ag{1})=2$} (A)
          (C) edge [bend left] node {$v(x\ag{2})=2$} (A);
    
    \end{tikzpicture}
    \end{center}    
\vspace*{-5mm}
\caption{Assumptions for the railway scenario}
\label{fig:assumptions-trains}
\vspace*{-2mm}
\end{figure}

\begin{example}\label{ex3:assumption}
Recall module $M\ag{C(2)}=(X,I,Q,T,\lambda,q_0)$ of the controller for 2 trains, with $Q=\{r,g\ag{1},g\ag{2}\}$.
We define four different assumptions about the behavior of the rest of the system, depicted graphically in Figure~\ref{fig:assumptions-trains}:
\begin{itemize2}
\item $A_0=(X,I,Q,T,\lambda,q_0,\{r\})$
\item $A_1=(X,I,Q,T,\lambda,q_0,\{g\ag{1}\})$
\item $A_2=(X,I,Q,T,\lambda,q_0,\{g\ag{2}\})$
\item $A_{012}=(X,I,Q,T,\lambda,q_0,\{r,g\ag{1},g\ag{2}\})$.
\end{itemize2}

Note that we can identify each valuation with an element of the set $\{0,1,2\}$, i.e., the value of the only variable $s$.
This way $A_0$ as well as $A_{012}$ accept all infinite words of the $\omega$-regular language
$L = (0(1|2))^\omega$, while $A_1$ and $A_2$ accept only proper subsets of this language,
namely $L\setminus(0(1|2))^*(02)^\omega$ and $L\setminus(0(1|2))^*(01)^\omega$.
\end{example}

\subsection{Guarantees}\label{sec:guar}

We say that a sequence $v=v_1v_2\ldots$ over $D^Y$ is
a \emph{curtailment} of a sequence $u=u_1u_2\ldots$ over $D^X$ (where $Y\subseteq X$)
if there exists an infinite sequence $c$ of indices $c_0<c_1<...$ with $c_0=0$
such that $\forall_i \forall_{c_i\leq k<c_{i+1}} v_i=u_k|_Y$.
We will denote a curtailment of $u$ to $D^Y$ by $u|_Y$ or $u|_Y^c$,
and use it to abstract away from irrelevant variables and the stuttering of states.

\begin{definition}[Guarantee]
Let $M\ag{1}, \ldots, M\ag{k}$ be pairwise asynchronous modules, and
$A=(X\ag{A},I\ag{A},Q\ag{A},T\ag{A},\lambda\ag{A},q\ag{A}_0,F\ag{A})$ be an assumption
with $X\ag{A}\subseteq X=\bigcup_{i=1}^k X\ag{i}$ and $I\ag{A}\subseteq I=\bigcup_{i=1}^k I\ag{i}$.

We say that $M=M\ag{1}|\ldots|M\ag{k}$ \emph{guarantees} the assumption $A$ (denoted $M\models A$)
if, for every infinite trace $\sigma$ of $M$ with $w\in (D^X)^\omega$ derived by $M$ with $\sigma$ and $u\in (D^I)^\omega$ admitted by $M$ with $\sigma$,
there exists a curtailment $w|_{X\ag{A}}^c$ ($c=c_1,c_2,\ldots$) accepted by $A$ with the computation $u_{c_1-1}|_{I\ag{A}}\;u_{c_2-1}|_{I\ag{A}}\;\dots$ .
\end{definition}
That is, every trace of $M$ must agree on the values of $X\ag{A}$ with some trace in $A$, modulo stuttering.

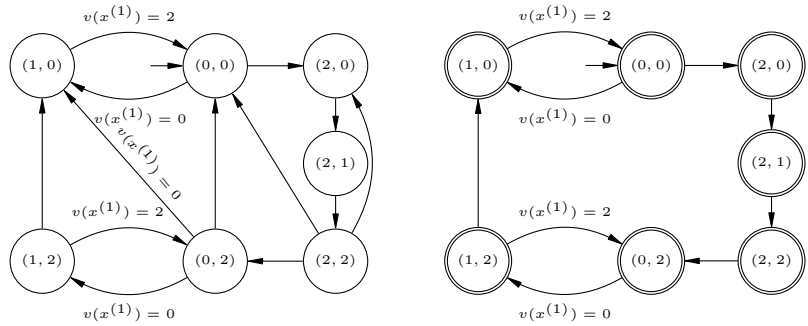
\begin{figure}[t]\centering
\begin{tikzpicture}[->,auto,>=arrow30,node distance=1.3cm,font=\tiny]
  \node[initial,state, initial text=]  (A)               {$(0,0)$};
  \node[state](B) [left of=A, xshift=-10mm]   {$(1,0)$};
  \node[state]          (C) [right of=A, xshift=3mm]  {$(2,0)$};
  \node[state]          (D) [below of=C, yshift=0mm]  {$(2,1)$};
  \node[state]          (E) [below of=D, yshift=0mm]  {$(2,2)$};
  \node[state]          (F) [left of=E, xshift=-3mm]   {$(0,2)$};
  \node[state](G) [left of=F, xshift=-10mm]   {$(1,2)$};
  
  \path (A) edge [bend left] node [xshift=.15cm] {$v(x\ag{1})=0$}      (B)
            edge             (C) 
        (B) edge [bend left] node {$v(x\ag{1})=2$}      (A)
        (F) edge [bend left] node {$v(x\ag{1})=0$}      (G)
        (G) edge [bend left] node [xshift=-.15cm] {$v(x\ag{1})=2$}      (F)
        (C) edge             (D) 
        (D) edge             (E) 
        (E) edge             (F) 
        (G) edge             (B) 
        (F) edge [right]     (A) 
        (F) edge             node [above, rotate = -50] {\;\;\;\;\;\;\;$v(x\ag{1})=0$} (B)
        (E) edge [bend right](C) 
        (E) edge             (A) 
        ;  
\end{tikzpicture}
\;\;\;\;\;\;\;\;
\begin{tikzpicture}[->,auto,>=arrow30,node distance=1.3cm,font=\tiny]
  \node[initial,state,accepting,initial text=]  (A)               {$(0,0)$};
  \node[state,accepting](B) [left of=A, xshift = -10mm]   {$(1,0)$};
  \node[state,accepting]          (C) [right of=A, xshift = 3mm]  {$(2,0)$};
  \node[state,accepting]          (D) [below of=C, yshift = 0mm]  {$(2,1)$};
  \node[state,accepting]          (E) [below of=D, yshift = 0mm]  {$(2,2)$};
  \node[state,accepting]          (F) [left of=E, xshift = -3mm]   {$(0,2)$};
  \node[state,accepting](G) [left of=F, xshift = -10mm]   {$(1,2)$};
  
  \path (A) edge [bend left] node {$v(x\ag{1})=0$}      (B)
            edge             (C) 
        (B) edge [bend left] node {$v(x\ag{1})=2$}      (A)
        (F) edge [bend left] node {$v(x\ag{1})=0$}      (G)
        (G) edge [bend left] node {$v(x\ag{1})=2$}      (F)
        (C) edge             (D) 
        (D) edge             (E) 
        (E) edge             (F) 
        (G) edge             (B) 
        ;
\end{tikzpicture}
\vspace*{-1mm}
\caption{Module $M\ag{C(2)}|M\ag{2}$ (left). The edges are labeled only if the value of $x\ag{1}$ is relevant.
Subsystem $M\ag{2}|A_{012}$ implementing strategy $\sigma$ (right).}
\vspace*{-2mm}
\label{fig:agv1}
\label{fig:composition-ass}
\end{figure}

\begin{example}\label{ex:modelassum}
Consider the system $M\ag{C(2)}|M\ag{1}|M\ag{2}$ presented in Example~\ref{ex:composition},
its subsystem $M\ag{C(2)}|M\ag{2}$ from Figure~\ref{fig:composition-ass},
and the assumption $A_{012}$ of Example~\ref{ex3:assumption}.

If we focus on the changes $s$, the following words can be derived:\
$(0(1|2))^\omega$ for the trains taking turns in the tunnel forever,\
$(0(1|2))^*01^\omega$ for the traces where the semaphore is stuck in state $(1,0)$ because it never receives that $v(x\ag{1})=2$,\
and $(0(1|2))^*02^\omega$ for ones that cycle forever in the right-hand part of $M\ag{C(2)}|M\ag{2}$.
In consequence, we have that $M\ag{C(2)}|M\ag{2}\models A_{012}$,
but not $M\ag{C(2)}|M\ag{2}\models A_{1}$.
\end{example}

It is possible to relate the traces of a subsystem with the traces of the entire system in such a way
that it is possible to verify locally defined formulas.

%====================================================================================
\section{\mbox{Assume-Guarantee Reasoning for 1ATL*}}\label{sec:agv-single-agents}
%====================================================================================

Now we propose our assume-guarantee schemes that decompose abilities of coalition $C$ into abilities of its subcoalitions, verified in suitable abstractions of their neighbor modules.

\subsection{Assume-Guarantee Rule for Strategies}\label{sec:agv-individual}

Let $\model$ be a system composed of asynchronous agents $(M\ag{1},R\ag{1}),\ \dots,\ (M\ag{n},R\ag{n})$.
By $N\ag{i}_1$, we denote the direct ``neighborhood'' of agent $i$, i.e., the set of agent indices $j$ such that $I_{M\ag{j}}\cap X_{M\ag{i}}\neq\es$ or $I_{M\ag{i}}\cap X_{M\ag{j}}\neq\es$.
By $N\ag{i}_k$, we denote the agents connected to $i$ in at most $k$ steps, i.e., $(N\ag{i}_{k-1}\cup\bigcup_{j\in N\ag{i}_{k-1}}N\ag{j}_1) \setminus \{i\}$.
Finally $\compos\ag{i}_k$ denotes the composition of all modules of $N\ag{i}_k$.
That is, if $N\ag{i}_k=\{a_1,...,a_m\}$ then $\compos\ag{i}_k=M\ag{a_1} | ... | M\ag{a_m}$.

Let $\psi_{i}$ be an \LTL formula (without ``next''), where atomic propositions are local valuations of variables in $M\ag{i}$.
Also, let $x\in\{\ir,\iR\}$.
The scheme is formalized through a sequence of rules $\bf{R_k}$ which rely on the behaviour of the neighbourhoods of coalition $C$, limited by ``distance'' $k$:

\[
\bf{R_k}\;\;\;
\begin{array}{c}
\forall_{i\in C}\; (M\ag{i} | A_i, R\ag{i}) \,\satisf[x]\, \coop{i} \psi_{i}\\
\forall_{i\in C}\; \compos\ag{i}_k \,\models\, A_i\\
\hline
(M\ag{1} | ... | M\ag{n},R\ag{1},\dots,R\ag{n}) \,\satisf[x]\, \coop{C} \bigwedge_{i\in C}\psi_{i}
\end{array}
\]

The main challenge in applying the scheme is to define the right assumptions and to decompose the verified formula.

\begin{example}\label{ex:rksound}
Recall the multi-agent system $\model$ presented in Example~\ref{ex:composition}, based on module $M\ag{C(2)}|M\ag{1}|M\ag{2}$.
We already argued that it satisfies
$\phi \equiv \coop{1,2} (GF p\ag{1})\land (GF p\ag{2})$
as well as $\phi' \equiv \coop{1,2} (GF q\ag{1})\land (GF q\ag{2})$,
for $p\ag{i}(x\ag{i})=1$ and
$q\ag{i}(s)=i$, cf.~Example~\ref{ex:logic}.
We will now see if the verification of the formulas can be decomposed using $\bf{R_k}$.

By Example~\ref{ex:modelassum} we know that $M\ag{C(2)}|M\ag{i}\models A_{012}$, where $A_{012}$ was an assumption defined in Example~\ref{ex3:assumption}.
It is easy to see that $M\ag{C(2)}\models A_{012}$.

Consider the extended module $M\ag{2}|A_{012}$, which is nothing but $M\ag{2}|M\ag{C(2)}$ with all the states marked as accepting.
Assume further that agent $2$ executes strategy $\sigma_2$ of Example~\ref{ex:logic}.
The resulting subsystem is presented in Figure~\ref{fig:agv1}.
Note that, if we focus on the values of variable $s$,
the $\omega$-regular language accepted by this automaton is $((01)|(0222(01)^*011))^\omega$, hence it periodically satisfies $p(\{s\})=1$.
In consequence, $\sigma_2$ can be used to demonstrate that
$(M\ag{2}|A_{012}, R\ag{2}) \satisf[\ir] \coop{2} GFq\ag{1}$, where $q\ag{1}(s)=1$.
Similarly, $(M\ag{1}|A_{012}, R\ag{1}) \satisf[\ir] \coop{1} GFq\ag{2}$, where $q\ag{2}(s)=2$.

As a result, we have decomposed formula $\phi'$ and constructed independent strategies for agents $1$ and $2$.
By the use of rule ${\bf R_1}$, we
conclude that $$(M,R\ag{C(2)},R\ag{1},R\ag{2}) \satisf[\ir] \coop{1,2}(GF q\ag{1})\land (GF q\ag{2}).$$

The situation for $\phi \equiv \coop{1,2} (GF p\ag{1})\land (GF p\ag{2})$ is drastically different.
We cannot use the analogous reasoning, because $\coop{i}GF p\ag{3-i}$ is not a local constraint for $M\ag{i}$.
There is a unique decomposition of $\phi$ into local constraints, but proving that $(M\ag{1}|A_{012}, R\ag{1}) \satisf[\ir] \coop{1} GFp\ag{1}$ fails,
as the system can get stuck in the state where $s$ equals $2$ or infinitely loop between the states where $s=2$ and $s=0$.
Changing the assumption would not help, since we cannot avoid the infinite exclusion of the considered train.
Thus, while the scheme can be used to derive that $\model \satisf[\ir] \phi'$, it cannot produce the (equally true) statement $\model \satisf[\ir] \phi$.
\end{example}

\subsection{Soundness and Incompleteness}\label{sec:soundness}

The following theorem says that, if each coalition member together with its assumption satisfies the decomposition of the formula, and its neighborhood satisfies the assumption, then the original verification task must return ``true.''

\begin{theorem}\label{t:rksound}
The rule $\bf{R_k}$ is sound.
\end{theorem}
\begin{proof}
Let $\forall_{i\in C}\; (M\ag{i} | A_i,R\ag{i}) \,\models_{x}\, \coop{i} \psi_i$
with (memoryless or perfect recall) imperfect information strategy $\sigma_i$
and $\forall_{i\in C}\; \compos\ag{i}_k \,\models\, A_i$.
Here and in the rest of the proof, $x\in\{ir,iR\}$.

Let us consider $M=M\ag{1} | ... | M\ag{n}$ such that $(M,R\ag{1},\ldots,R\ag{n})\models_{x} \coop{C} \psi_i$ and fix its joint strategy
$\sigma$ for coalition $C$, where $\sigma(i)=\sigma_i$ for every $i\in C$.

We will prove the soundness by contradiction.
Suppose that for every (memoryless or perfect recall) imperfect information joint strategy
there exists an infinite word which implements
this joint strategy, but do not satisfy
$\bigwedge_{i\in C}\psi_i$, i.e. there exists $j\in C$ such that
$w$ does not satisfy $\psi_j$.
Let $w=q_0q_1\ldots$ be such a word for the strategy $\sigma$ and fix $j$.

Let us consider $M\ag{j} | A_j$, where $X_{M\ag{j}}$ and $X_{A_j}$ are internal variables of
$M\ag{j}$ and $A_j$, appropriately.
By the construction and the presumption that
$(M\ag{j} | A_j,\R\ag{j}) \,\models_{x}\, \coop{j} \psi_j$
we get that every infinite word
over $X_{M\ag{j}}\cup X_{A_j}$ which implement (memoryless or perfect recall) imperfect information
strategy $\sigma_j$ satisfy $\psi_j$.

However, the assumption $A_j$
is guaranteed by $\compos\ag{j}_k$,
hence for a word derived by $\compos\ag{j}_k$ we have its curtailment accepted by $A_j$.
Moreover, every word accepted by $M\ag{j} | A_j$ is a curtailment of a word derived by $M$,
and, in particular, $w$ is such a word.
However, 
there exists a curtailment $w|_{X_{M\ag{j}}\cup X_{A_j}}$
which satisfy strategy $\sigma_j$ but is not accepted by $M\ag{j} | A_j$,
which gives an obvious contradiction with
$(M\ag{j} | A_j,R\ag{j}) \,\models_{x}\, \coop{i} \psi_j$.

The obtained contradiction shows that there exists a joint strategy $\sigma$ for the entire model and
$(M\ag{1} | ... | M\ag{n},R\ag{1},\ldots,R\ag{n}) \models_{x} \coop{C} \bigwedge_{i\in C}\psi_i$, which concludes the proof.
\end{proof}

Unfortunately, there does not always exist $k<n$ for which the rule $\bf{R_k}$ is complete,
even in a very weak sense, where we only postulate the \emph{existence} of appropriate assumptions.

\begin{theorem}
The scheme consisting of rules $\set{\bf{R_k} \mid k\in\Nat}$ is in general not complete.
\end{theorem}
\begin{proof}
Follows directly from Example~\ref{ex:rksound}.
\end{proof}

%====================================================================================
\subsection{Coalitional Assume-Guarantee Verification}\label{sec:agv-coalitions}
%====================================================================================

In Section~\ref{sec:soundness}, we showed that achievable coalitional goals may not decompose into achievable individual subgoals.
As a result, the scheme proposed in Section~\ref{sec:agv-individual} is incomplete.
A possible way out is to allow for assume-guarantee reasoning about joint strategies of subcoalitions of $C$.
We implement the idea by partitioning the system into smaller subsystems and allowing to explicitly consider the cooperation between coalition members.

Again, let $\model=(M\ag{1},R\ag{1}),\ \dots,\ (M\ag{n},R\ag{n})$ be a system composed of asynchronous agents .
Moreover, let $\{P_1, \ldots, P_k : P_i\subseteq\{1,2,\ldots,n\}\}$,
be a partitioning of coalition $C$,
and let $\overline{C}=\{i : i\notin C\}=Ag\setminus C$ be the set of opponents of $C$.
By $\model\ag{P_i}$ we denote the system composed of all the agents in $P_i = \{i_1,\ldots,i_s\}$, i.e., $(M\ag{P_i}=M\ag{i_1}|\ldots|M\ag{i_s}, R\ag{i_1},\ldots,R\ag{i_s})$.
$\model\ag{\overline{C}}$ is defined analogously.

We extend the notion of neighbourhood to sets of agents as follows:\
\begin{itemize}
\item $N^{P_i}_1=(\bigcup_{i\in P_i}N\ag{i}_1)\setminus P_i$,\
$N^{P_i}_k=(N^{P_i}_{k-1}\cup\bigcup_{j\in N^{P_i}_{k-1}}N\ag{j}_1) \setminus P_i$ for $k>1$,\ 
\item $\compos^{P_i}_k=M\ag{x_1} | ... | M\ag{x_s}$ for $N^{P_i}_k=\{x_1,...,x_s\}$.
\end{itemize}

Let $x\in \{\ir,\iR\}$. The generalized assume-guarantee rule is defined below:

\[
\bf{Part^P_k}\;\;\;
\begin{array}{c}
\forall_{P_i\in P}\; (M\ag{P_i} | A_i,R\ag{i_1},\ldots,R\ag{i_s}) \,\models_{x}\, \coop{P_i} \bigwedge_{j\in P_i}\psi_j\\
\forall_{P_i\in P}\; \compos^{P_i}_k \,\models\, A_i\\
\hline
(M\ag{1} | ... | M\ag{n},R\ag{1},\ldots,R\ag{n}) \models_{x} \coop{C} \bigwedge_{i\in C}\psi_i
\end{array}
\]

As it turns out, the new scheme is sound, conservative with respect to enlarging the neighborhood, and complete.

\begin{theorem}\label{t:partsound}
The rule $\bf{Part^P_k}$ is sound.
\end{theorem}
\begin{proof}
Intuitively, we can proceed similarly to the proof of Theorem~\ref{t:rksound}.
Note that each component $P_i$ can be seen as single composed module with a imperfect information
strategy (memoryless or with perfect recall) being a joint
strategy for the subset of coalition $C$ which is the component $P_i$.
Moreover, we can take instead of $\compos^{P_i}_k$ the union $U$ of all the components $P_j$
(and possibly $\overline{C}$) which
intersection with $\compos^{P_i}_k$ is non-empty.
It is easy to see, that if $\compos^{P_i}_k \,\models\, A_i$ then also $U\,\models\, A_i$.

This way we could fix the strategy for coalition $C$, deduce that every infinite word
as a composition of strategies $\sigma_{P_i}$ for its parts $(C\cap P_i)_{P_i\in P}$, and
deduce that for every word $w$ which do not satisfy $\bigwedge_{i\in C}\psi_i$ there
exists a single component $P_j$ containing $M_i$ such that $\psi_j$ would not be satisfied
for any curtailment of $w$, while one of them implements strategy $\sigma_{P_i}$
being at the same time accepted by $M\ag{P_i}|A_i$.
\end{proof}

\begin{proposition}\label{p:enlarge}
If $\compos^{P_i}_k \,\models\, A_i$ then $\compos^{P_i}_{k+1} \,\models\, A_i$.
\end{proposition}
\begin{proof}
Let $N^{P_i}_{k+1} = \{i_1,\ldots,i_t\}$, $M=M\ag{i_1}|\ldots M\ag{i_t}$,
$N^{P_i}_{k} = \{j_1,\ldots,j_{t'}\}$ and $M'=M\ag{j_1}|\ldots M\ag{j_{t'}}$.
Let us consider an infinite trace $\sigma$
of $M$, with $w\in (D^X)^\omega$ and $u\in(D^I)^\omega$
and $\sigma'=w_1|_{X'}\,(u_1|_{I'\cap I}\cup w_1|_{I'\cap X})\,w_2|_{X'}\ldots$.
Note that one of the curtailments of a word $w_1|_{X'}w_2|_{X'}\ldots$
is derived by $M'$, and thus its curtailment is accepted by $A_i$.
\end{proof}

\begin{theorem}\label{t:complete}
There exist a partition set $P$ and $k\leq n$ such that the rule $\bf{Part^P_k}$ is complete.
\end{theorem}
\begin{proof}
Straightforward, as we can take $k=n$ and singleton partition $P=\{P_1\}$, where
$A_1$ is an automaton constructed on the base of the system $M\ag{\overline{C}}$,
where all the states are accepting ones
(hence $\compos^k_{P_1}\models A_1$ as every word accepted by $A_1$
is derived with a trace of $M_{\overline{C}}$).

Hence $(M\ag{P_1} | A_1,R\ag{i_1},\ldots,\R\ag{i_s}) \,\models_{x}\, \coop{P_1} \bigwedge_{j\in P_1}\psi_j$
is just an equivalent formulation of
$(M\ag{1} | ... | M\ag{n},R\ag{1},\ldots,R\ag{n}) \models_{x} \coop{C} \bigwedge_{i\in C}\psi_i$,
for $x\in\{ir,iR\}$.
\end{proof}

\begin{remark}[Complexity]
The assume-guarantee schemes provide (one-to-many) reductions of the model checking problem. The resulting verification algorithm for \ATLs[ir] is \PSPACE-complete with respect to the size of the coalition modules, the size of the assumptions, and the length of the formula. In the very worst case (i.e., as the assumptions grow), this becomes \PSPACE-complete w.r.t.~the size of the global model, i.e., no better than ordinary model checking for \ATLs with memoryless strategies.
On the other hand, our method often allows to decompose the verification of the huge global model of the system to several smaller cases. For many systems one can propose assumptions that are exponentially smaller than the size of the full model, thus providing an exponential gain in complexity.

Note also that the first scheme provides a model checking algorithm for \ATLs[iR] that is \EXPTIME-complete with respect to the size of the coalition modules, the size of the assumptions, and the length of the formula, i.e., an incomplete but decidable algorithm for the generally undecidable problem.
\end{remark}

\subsection{Verification of Nested Strategic Operators}\label{sec:nested}

So far, we have concentrated on assume-guarantee specification of formulas without nested strategic modalities.
Here, we briefly point out that the schemes $\bf{R_k}$ and $\bf{Part^P_k}$ extend to the whole language of \ATLsX through the standard recursive model checking algorithm that verifies subformulas bottom-up.
Such recursive application of the method to the verification of $\model \models \phi$ proceeds as follows:
\begin{itemize}
\item For each strategic subformula $\phi_j$ of $\phi$, do assume-guarantee verification of $\phi_j$ in $\model$, and label the states where $\phi_j$ holds by a fresh atomic proposition $\prop{p_j}$;
\item Replace all occurrences of $\phi_j$ in $\phi$ by $\prop{p_j}$, and do assume-guarantee verification of the resulting formula in $\model$.
\end{itemize}

The resulting algorithm is sound, though there is the usual price to pay in terms of computational complexity.
The main challenge lies in providing decompositions of \LTL objectives for multiple strategic formulas, as well as multiple Büchi assumptions (one for each subformula).
A refinement of the schemes for nested strategic abilities is planned for future work.

%====================================================================================
\section{Case Study and Experiments}\label{sec:experiments}
%====================================================================================

In this section, we present an experimental evaluation of the assume-guarantee verification schemes of Section~\ref{sec:agv-single-agents}.
As the benchmark, we use a variant of the factory scenario from~\cite{Kurpiewski19embedded}, where a coalition of logistic robots cooperate to deliver packages from
the production line to the storage area.

\subsection{Experiments: Monolithic vs. Assume-Guarantee Verification}

\para{Decomposition to Individual Strategies.}
In the first set of experiments, we verified the formula

\vspace{1mm}
\centerline{
$\psi\ \equiv\ \coop{R}(\bigwedge_{r\in R}\prop{energy_r}>0) \Until \prop{delivered}$
}

\vspace{1mm}
\noindent
expressing that the coalition of robots $R$ can maintain their energy level above zero until at least one package is delivered to the storage area.
Guessing that the first robot has enough energy to deliver a package on his own, we can decompose the formula as the conjunction of the following components:

\vspace{1mm}
\centerline{
$\psi_d\ \equiv\ \coop{r_1}\Sometm \prop{delivered}$,
\hspace{10mm}
$\psi_e\ag{i}\ \equiv \coop{r_i}\Always \prop{energy_{r_i}}>0$,\quad $i\in R$.
}

\vspace{1mm}
\noindent
Note that, if $\psi_d\land\bigwedge_{i>1}\psi_e\ag{i}$ is true, then $\psi$ must be true, too.

The experiments used the first (incomplete) scheme of assume-guarantee verification.
The results are presented in Table~\ref{tab:res}.
The first {column} describes the configuration of the model, i.e., the number of robots, locations in the factory, and the initial energy level.
Then, we report the performance of model checking algorithms that operate on the explicit model of the whole system.
The running times are given in seconds.
\emph{DFS} is a straightforward implementation of depth-first strategy synthesis. \emph{Apprx} refers to the (sound but incomplete) method of fixpoint-approximation in~\cite{Jamroga19fixpApprox-aij}; besides the time, we also report if the approximation was conclusive.

\para{Coalitional Assume-Guarantee Verification.}
For the second set of experiments, the robots were divided in two halves, initially located in different parts of the factory.
We verified the following formula:

\vspace{1mm}
\centerline{
$\psi\ \equiv\ \coop{R}\Sometm\Always(\prop{\bigwedge_{i\in \{1,2,...,n/2\}}(delivered_i \vee delivered_{i+n/2})})$,
}
\vspace{1mm}
\noindent
expressing that the coalition of robots can delivered at least one package per pair to the {storage area}.
Depending on the initial energy level of robots, the storage may not be reachable from the {production line}.
That means that the robots must work in pairs to deliver the packages.
We use this insight to decompose the verification into the following formulas:

\vspace{1mm}
\centerline{
$\psi\ag{i}\ \equiv\ \coop{r_i, r_{i+n/2}}\Sometm\Always(\prop{delivered_i \vee delivered_{i+n/2}})$.
}

\vspace{1mm}
\noindent
The results are presented in Table~\ref{tab:res}.

\begin{table}[t]
\begin{adjustbox}{width=1.1\columnwidth,center}
\begin{tabular}{|c|cH|c|c|cH|c|c|c@{\quad}c|c|cH|c|c|cH|c|c|}
\cline{1-9} \cline{12-20}
\multirow{2}{*}{\textbf{Conf}}  & \multicolumn{4}{c|}{\textbf{Monolithic verif.}}                & \multicolumn{4}{c|}{\textbf{Ass.-guar. verif.}}   & & & \multirow{2}{*}{\textbf{Conf}}  & \multicolumn{4}{c|}{\textbf{Monolithic verif.}}       & \multicolumn{4}{c|}{\textbf{Ass.-guar. verif.}}  \\ \cline{2-9} \cline{13-20}
                                & \textbf{\#st} & \textbf{\#tr} & \textbf{DFS} & \textbf{Apprx}          & \textbf{\#st} & \textbf{\#tr} & \textbf{DFS} & \textbf{Apprx} & & &                                 & \textbf{\#st} & \textbf{\#tr} & \textbf{DFS} & \textbf{Apprx} & \textbf{\#st} & \textbf{\#tr} & \textbf{DFS} & \textbf{Apprx} \\ \cline{1-9} \cline{12-20}
2,2,2                           & 8170          & 25313         & <0.01         & 0.6/No                  & 1356          & 2827          & <0.01         & <0.01/Yes    & & & 2,3,1                           & 522           & 1352          & <0.01         & <0.01/No        & 522           & 1352          & <0.01         & <0.01/No     \\ \cline{1-9} \cline{12-20}
2,3,3                           & 1.1e5         & 4.4e5         & 0.02         & 13/No                   & 9116          & 2.1e4         & <0.01         & 0.5/Yes       & & & 2,4,2                           & 3409          & 9696          & <0.01         & <0.01/No        & 3409          & 9696          & <0.01         & <0.01/No  \\ \cline{1-9} \cline{12-20}
3,2,2                           & 5.5e5         & 3.6e6         & \multicolumn{2}{c|}{timeout}           & 2.7e4         & 8.4e4         & <0.01         & 3/Yes         & & & 4,3,1                           & \multicolumn{4}{c|}{memout}                                   & 4.8e4         & 2.8e5         & <0.01         & 4/No  \\ \cline{1-9} \cline{12-20}
3,3,3                           & \multicolumn{4}{c|}{memout}                                            & 4.4e5         & 1.6e6         & <0.01         & 58/Yes        & & & 6,3,1                           & \multicolumn{4}{c|}{memout}                                   & 5.8e5         & 4.2e6         & 0.36         & 42/No  \\ \cline{1-9} \cline{12-20}
4,2,2                           & \multicolumn{4}{c|}{memout}                                            & 5.2e5         & 2.1e6         & \multicolumn{2}{c|}{timeout}  & & & 8,3,1                           & \multicolumn{4}{c|}{memout}                                   & \multicolumn{4}{c|}{timeout}                         \\ \cline{1-9} \cline{12-20}
\end{tabular}
\end{adjustbox}
\vspace{1mm}
\caption{Results of assume-guarantee verification, scheme $\bf{R_k}$ (left), scheme $\bf{Part^P_k}$ (right).}
\label{tab:res}
\vspace{-8mm}
\end{table}

\para{Discussion of Results.}
The experimental results show that assume-guarantee schemes presented here enables to verify systems of distinctly higher complexity than model checking of the full model.
We have also conducted analogous experiments on the Simple Voting scenario of~\cite{Jamroga19fixpApprox-aij}, with very similar results; we do not report them here due to lack of space.

Interestingly, Table~\ref{tab:res} shows that the application of incomplete assume-guarantee scheme to fixpoint approximation (in itself an incomplete method of model checking) often turns inconclusive verification into conclusive one.
This is because fixpoint approximation works rather well for individual abilities, but poorly for proper coalitions~\cite{Jamroga19fixpApprox-aij}.
Rule $\bf{R_k}$ decomposes verification of coalitional abilities (very likely to resist successful approximation) to model checking individual abilities (likely to submit to approximation).
It is not true in the case of the second experiment, as this time we did not reduce the tested coalitions to singleton ones.

%====================================================================================
\section{Conclusions}\label{sec:conlusions}
%====================================================================================

In this paper we propose two schemes for assume-guarantee verification of strategic abilities.
Importantly, they are both sound for the memoryless as well as perfect recall semantics of abilities under imperfect information.
Moreover, the second scheme is complete (albeit in a rather weak sense).
The experiments show that both schemes can provide noticeable improvement in verification of large systems consisting of asynchronous agents with
independent goals.
Note also that the scheme $\bf{R_k}$ provides an (incomplete) reduction of the undecidable model checking problem for coalitions with perfect recall to decidable verification of individual abilities.

Clearly, the main challenge is to formulate the right assumptions and to decompose the verified formula.
In the future, we would like to work on the automated generation of assumptions.
The first idea is to obtain a larger granularity of the global model by decomposing agents into even smaller subsystems (and recomposing some of them as assumptions).
This can be combined with abstraction refinement of the assumptions in case they are still too complex.
We also plan to extend the notion of assumptions to capture the agents' knowledge about the strategic abilities of their coalition partners.
Positive results in that direction would significantly increase the applicability of assume-guarantee schemes for model checking of asynchronous MAS.

%====================================================================================
\subsubsection{Acknowledgement}
The work was supported by NCBR Poland and FNR Luxembourg under the PolLux/FNR-CORE project STV (POLLUX-VII/1/2019), and the CHIST-ERA grant CHIST-ERA-19-XAI-010 by NCN Poland (2020/02/Y/ST6/00064). The work of Damian Kurpiewski was also supported by the CNRS IEA project MoSART.
%====================================================================================

\bibliographystyle{splncs04}
\bibliography{biblio}

\end{document}